\documentclass{article}

\usepackage[UKenglish]{babel}
\usepackage[utf8]{inputenc}
\usepackage[T1]{fontenc}
\usepackage{amsfonts,amsmath,amssymb}
\usepackage{amsthm}
\usepackage[usenames]{color} 
\usepackage{xspace}
\usepackage[svgnames]{xcolor}
\usepackage{bookman}
\usepackage{bbold}
\usepackage{graphicx}
\usepackage{framed}
\usepackage{subcaption}
\usepackage[norelsize,ruled]{algorithm2e}

\usepackage[margin=2.5cm]{geometry}
\newtheorem{lemma}{Lemma}

\newcommand{\set}[1]{\left\{ #1 \right\}}

\newcommand{\N}{\mathbb{N}}

\newtheorem{theorem}{Theorem}
\newtheorem{corollary}{Corollary}

\newtheorem{definition}{Definition}

\newtheorem{remark}{Remark}

\newcommand{\game}{\mathcal{G}} 
\newcommand{\A}{\mathcal{A}} 

\newcommand{\Win}{\textsc{Win}} 
\newcommand{\Lose}{\textsc{Lose}} 

\newcommand{\VE}{V_E} 
\newcommand{\VA}{V_A} 

\newcommand{\WE}{W_E} 
\newcommand{\WA}{W_A} 

\newcommand{\play}{\pi} 

\newcommand{\Pre}{\text{Pre}} 
\newcommand{\Lift}{\text{Lift}} 

\newcommand{\AllEvenCycles}{\text{AllEvenCycles}} 
\newcommand{\AllOddCycles}{\text{AllOddCycles}} 

\newcommand{\Safe}{\text{Safe}}
\newcommand{\Reach}{\text{Reach}}

\newcommand{\Parity}{\text{Parity}} 

\newcommand{\F}{\mathcal{F}} 
\newcommand{\G}{\mathcal{G}}

\begin{document}


\title{An Optimal Value Iteration Algorithm for Parity Games}
\author{Nathana{\"e}l Fijalkow}
\date{CNRS, LaBRI, Bordeaux, France\\Alan Turing Institute, London, United Kingdom\\University of Warwick, United Kingdom}


\maketitle

\begin{abstract}
The quest for a polynomial time algorithm for solving parity games gained momentum in 2017 when two different quasipolynomial time algorithms were constructed. In this paper, we further analyse the second algorithm due to Jurdzi{\'n}ski and Lazi{\'c} and called the succinct progress measure algorithm. It was presented as an improvement over a previous algorithm called the small progress measure algorithm, using a better data structure.

The starting point of this paper is the observation that the underlying data structure for both progress measure algorithms are (subgraph-)universal trees. We show that in fact any universal tree gives rise to a value iteration algorithm {\`a} la succinct progress measure, and the complexity of the algorithm is proportional to the size of the chosen universal tree. We then show that both algorithms are instances of this generic algorithm for two constructions of universal trees, the first of exponential size (for small progress measure) and the second of quasipolynomial size (for succinct progress measure).

The technical result of this paper is to show that the latter construction is asymptotically tight: universal trees have at least quasipolynomial size. This suggests that the succinct progress measure algorithm of Jurdzi{\'n}ski and Lazi{\'c} is in this framework optimal, and that the polynomial time algorithm for parity games is hiding someplace else.

\end{abstract}

\section{Introduction}

The notion of parity games is fundamental in the study of logic and automata.
Most often fundamental notions have very simple definitions and they clearly capture a key aspect of the general problem of interest.
This cannot be said of parity games: the definition takes a bit of time to digest and once understood it is not clear how central it may be.
Indeed, it took years, if not decades, to formulate the right notion to look at.

\vskip1em
Parity games first appeared in the context of automata over infinite trees.
The first and natural idea to define automata over infinite objects is to have so-called Muller conditions, 
where to determine whether a run is accepted one considers which states appear infinitely often.
One can develop a rich theory relating automata and logic over infinite trees using Muller automata, 
but some properties are very hard to prove, as witnessed for instance by the technical ``tour de force'' of Rabin 
for proving the decidability of monadic second-order logic~\cite{Rabin69}.
The parity condition appeared in an effort to better understand this proof, and its importance became manifest: 
working with parity automata rather than Muller automata gives an arguably short and understandable proof of Rabin's celebrated result.
It was introduced independently by Mostowski~\cite{Mostowski84,Mostowski91}, who called them ``Rabin chain condition'',
and Emerson and Jutla~\cite{EJ91}.

The crucial property making the technical developments easier is the positional determinacy of parity games, which is the key result
used in many constructions for parity automata.
In a precise sense, one can show that the parity objectives form the largest class of Muller objectives enjoying positional determinacy,
a result due to Zielonka~\cite{Zielonka98}, see also~\cite{DJW97}.

\vskip1em
The main algorithmic problem about parity games is to solve them, i.e. to construct an algorithm taking as input a parity game
and determining whether the first player Eve has a winning strategy.
A strong motivation for constructing efficient algorithms for this problem is the works of Emerson and Jutla~\cite{EJ91},
who showed that solving parity games is linear-time equivalent to the model-checking problem for modal $\mu$-calculus.
This logical formalism is an establised tool in program verification, and a common denominator to a wide range of modal, temporal and fixpoint logics
used in various fields.

\vskip1em
The literature on algorithms for solving parity games is vast.
Up until 2017, the best algorithms were subexponential.
Two breakthroughs came in 2017: first the succinct counting algorithm of Calude et al~\cite{CJKL017}, 
and then the succinct progress measure of Jurdzi{\'n}ski and Lazi{\'c}~\cite{JL17}, 
both solving parity games in quasipolynomial time, more precisely in $n^{O(\log(d))}$,
for $n$ the number of vertices and $d$ the number of priorities.

\vskip1em
The aim of this paper is to further analyse the second algorithm and to relate it to the notion of universal trees.
Under this new light, we construct a mildly improved algorithm and prove its optimality within this framework.



\section{Definitions}
\label{sec:parity}

The \textit{arena} is the place where the game is played:
the first component is a directed graph given by a set $V$ of vertices and a set $E \subseteq V \times V$ of edges.
Additionally, an arena features two sets $\VE$ and $\VA$ of vertices such that $V = \VE \uplus \VA$:
the set $\VE$ is the set of vertices controlled by Eve, and the set $\VA$ is those controlled by Adam.
We represent vertices in $\VE$ by circles, and vertices in $\VA$ by squares, and also say that $v \in \VE$
belongs to Eve, and similarly for Adam.
The relevant algorithmic parameters are the number $n$ of vertices and $m$ of edges of the arena.

\vskip1em
The interaction between the two players consists in moving a token on the vertices of the arena.
It is initially on the vertex $v_0$, starting the game.
When the token is in some vertex, the player who controls the vertex chooses an outgoing edge
and pushes the token along this edge to the next vertex.
To ensure not to get stuck we usually, although not always, assume that from any vertex there is an outgoing edge.
The outcome of this interaction is the infinite sequence of vertices traversed by the token, called a \textit{play}.
Plays are usually written $\play$, with $\play_i$ the $i$\textsuperscript{th} vertex of $\play$ (indexed from~$0$),
and $\play_{\le i}$ the prefix up to length $i$.
We let $V^\omega$ denote the set of plays, \textit{i.e.} infinite sequences of vertices, and $V^*$ the set of paths,
\textit{i.e.} finite sequences of vertices.

\vskip1em
A \emph{strategy} for a player is a full description of his or her moves in all situations.
Formally, a strategy is a function $\sigma : V^* \to E$ mapping any path to an edge.
Traditionally, strategies for Eve are written $\sigma$, and strategies for Adam are written~$\tau$.
We say that a play $\play$ is consistent with a strategy $\sigma$ for Eve if
for all $i \in \N$ such that $\play_i \in \VE$, we have $\sigma(\play_i) = (\play_i,\play_{i+1})$.
Once an initial vertex $v_0$, a strategy $\sigma$ for Eve, and a strategy $\tau$ for Adam have been fixed, 
there exists a unique play starting from $v_0$ and consistent with both strategies, written $\play^{v_0}_{\sigma,\tau}$.

\vskip1em
So far we defined the rules for playing (the arena), the means to play (the strategy), it remains to explain the goals to achieve (the objective).

We fix a set $C$ of colours and equip the arena with a function $c : V \to C$ mapping vertices to colours.
An objective $\Omega$ is a subset $\Omega \subseteq C^\omega$, which we interpret as the set of winning plays.
Recall that a play is an element of $V^\omega$, so thanks to the mapping $c : V \to C$,
it induces an element of $C^\omega$. 
If the element of $C^\omega$ induced by $\play$ is in $\Omega$, we say that $\play$ satisfies $\Omega$, or that $\play$ is winning.
A strategy $\sigma$ for Eve is winning from $v_0$ if for all strategies $\tau$ for Adam, the play $\play^{v_0}_{\sigma,\tau}$ is winning.
We sometimes say that the strategy $\sigma$ ensures $\Omega$, and that Eve wins from $v_0$.

\begin{definition}[Games]
Let $C$ be a set.
\begin{itemize}
	\item An arena $\A$ is a tuple $(V,E,\VE,\VA,c)$ where $(V,E)$ is a directed graph with $V = \VE \uplus \VA$
	and $c : V \to C$ maps vertices to colours.
	\item An objective $\Omega$ is a subset $\Omega \subseteq C^\omega$.
\end{itemize}
A game $\game$ is a pair $(\A,\Omega)$ where $\A$ is an arena and $\Omega$ an objective.
The generic algorithmic question we address is the following decision problem, later refered to as ``solving the game'':
\begin{framed}
\begin{tabular}{ll}
\textbf{INPUT}: & A game $\game$ and an initial vertex $v_0$\\
\textbf{QUESTION}: & Does Eve win from $v_0$?
\end{tabular}
\end{framed}
\end{definition}

We let $\WE(\game)$ denote the set of vertices from which Eve has a winning strategy in the game~$\game$.
When the arena is clear from the context and we consider different objectives over the same arena, 
we write $\WE(\Omega)$ for the set of vertices from which Eve has a strategy ensuring $\Omega$.

\vskip1em
We now define the \textit{parity} objectives.
Let $d \in \N$ be an even number defining the number of priorities.
The parity objective with parameter $d$ use the set of colours $\set{1,2,\ldots,d}$, which are referred to as priorities, and is defined by
\[
\Parity = \set{\play \in V^\omega \left| \begin{array}{l} \text{the largest priority appearing} \\ \text{infinitely often in } \play \text{ is even} \end{array} \right.}.
\]
We illustrate the definition on two examples.
\[
\begin{array}{c}
1\ 2\ 4\ 7\ 5\ 7\ 5\ 3\ 6\ 3\ 6\ 3\ 6\ 3\ 6\ \cdots \in \Parity \\
2\ 2\ 2\ 4\ 1\ 7\ 5\ 3\ 3\ 3\ 3\ 3\ 3\ 3\ 3\ \cdots \notin \Parity
\end{array}
\]
In the first play the two priorities which appear infinitely often are $3$ and $6$, and the largest one is $6$, which is even,
and in the second play the only priority which appears infinitely often is $3$ and it is odd.
Figure~\ref{fig:parity_game_example} presents an example of a parity game. 
The priority of a vertex is given by its label.
\begin{figure}[ht]
\centering
\includegraphics[scale=.3]{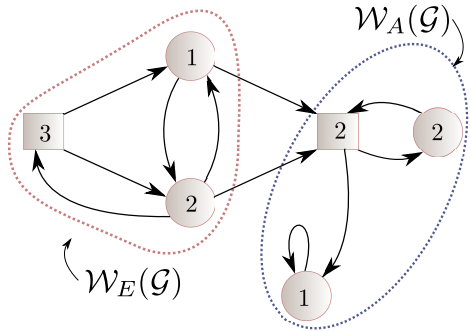}
\caption{An example of a parity game.}
\label{fig:parity_game_example}
\end{figure}

\vskip1em
This paper continues a long line of work aiming at constructing efficient algorithms for solving parity games.
Before starting the technical developments, let us discuss two important properties of parity games:
\begin{itemize}
	\item They are \emph{determined}, meaning that from any vertex, either Eve has a winning strategy or Adam has a winning strategy,
	which symbolically reads 
	\[
	\WE(\Parity) \cup \WA(\Parity) = V;
	\]
	\item The are \emph{positionally determined}, meaning that if Eve has a winning strategy, then she has a positional one,
	\textit{i.e.} of the form $\sigma : V \to E$. Such a strategy is called positional, sometimes memoryless, because it picks the next move
	only considering the current position, forgetting about the path played so far.	
\end{itemize}
The determinacy of parity games follows from very general topological theorems as for instance Martin's theorem~\cite{Martin75}.
The positional determinacy is due to Emerson and Jutla~\cite{EJ91}.

\vskip1em
\noindent\textbf{Organisation of the paper.}
In Section~\ref{sec:zielonka} we define signatures and show how analysing Zielonka's algorithm yields the existence of signatures.
This result is used in Section~\ref{sec:algorithm} for constructing and proving the correctness of the generic value iteration algorithm.
Here generic means that the algorithm is parameterised by an underlying data structure called a universal tree.
We explain how both the small progress measure and the succinct progress measure algorithms are instances of this framework.
Section~\ref{sec:bounds} shows asymptotically tight bounds on the size of universal trees.

This paper is self-contained, in particular does not rely on two properties mentioned above (determinacy and positional determinacy).
More accurately, we obtain them both in the next section as by-products of our analysis of Zielonka's algorithm.



\section{Signatures and Zielonka's algorithm}
\label{sec:zielonka}

In this section we revisit the notion of signatures for parity games,
which will be the key ingredient for the correctness proof of the generic value iteration algorithm in the next section.

The notion of signature was proposed by B{\"u}chi~\cite{Buchi83} and independently by Streett and Emerson~\cite{SE84}.
Emerson and Jutla~\cite{EJ91} used them to give a proof of positional determinacy for parity games.

\subsection*{Signatures}

We work with tuples in $[0,n]^{d/2}$ which we index by odd priorities in $[1,d]$.
For instance for $d = 8$, an example of a tuple $x$ is
\[
x = (\underbrace{2}_7, \underbrace{2}_5,\underbrace{3}_3,\underbrace{0}_1).
\]
We order tuples lexicographically, with the largest priority being the most important,
so we have $(2,2,3,0) >_{\text{lex}} (1,5,5,5)$.
For a priority $p$ and $x$ a tuple in $[0,n]^{d/2}$, we write $x_{\ge p}$ for the tuple restricted to priorities larger
than or equal to $p$. 
For the tuple $x$ above, we have $x_{\geq 5} = (2,2)$ and $x_{\geq 2} = (2,2,3)$.

We consider functions $\mu : V \to [0,n]^{d/2} \cup \set{\top}$.
It induces a set of orders on vertices called the $p$-orders: for $p$ a priority and $v,v'$ two vertices, 
we write $\mu(v) \ge_p \mu(v')$ if $\mu(v)_{\geq p} \ge_{\text{lex}} \mu(v')_{\geq p}$
and add $\top$ as the largest element for all $p$-orders.

\begin{definition}[Signatures]
Let $\game$ be a parity game with $n$ vertices and $d$ priorities.
A function $\mu : V \to [0,n]^{d/2} \cup \set{\top}$ is called a signature if it satisfies the following two properties:
\begin{itemize}
	\item If $v \in V_E$ has priority $p$, then there exists $(v,v') \in E$ such that $\mu(v) \ge_p \mu(v')$, 
	and the inequality is strict if $p$ is odd;
	\item If $v \in V_A$ has priority $p$, then for all $(v,v') \in E$ we have $\mu(v) \ge_p \mu(v')$, 
	and the inequality is strict if $p$ is odd.
\end{itemize}
\end{definition}

The notion of signatures is best explained by the following lemma, which reads: 
a signature is both a strategy for Eve and a proof that it is winning.

\begin{lemma}\label{lem:signature_correctness}
For all parity games with $n$ vertices and $d$ priorities, 
if $\mu : V \to [0,n]^{d/2} \cup \set{\top}$ is a signature and for $v \in V$ we have $\mu(v) \neq \top$, then Eve wins from $v$.
\end{lemma}

\begin{proof}
We first observe that $\mu$ induces a (positional) strategy $\sigma$ on vertices $v \in \VE$ such that $\mu(v) \neq \top$.
Indeed, for $v \in \VE$ of priority $p$, by definition there exists $(v,v') \in E$ such that $\mu(v) \ge_p \mu(v')$, define $\sigma(v) = v'$.

We claim that $\sigma$ is winning on the set of vertices $v \in V$ such that $\mu(v) \neq \top$.
To this end, consider a cycle
\[
v_1, v_2, \ldots, v_k, v_1
\] 
consistent with $\sigma$,
and assume for the sake of contradiction that the largest priority in the cycle is odd.
Without loss of generality we assume $v_1$ has the largest priority in the cycle, say $p$.
We then have by definition of a signature, and noting $c(v_i)$ the priority of $v_i$:
\[
\mu(v_1) >_p \mu(v_2) \ge_{c(v_2)} \mu(v_3) \ge_{c(v_3)} \cdots \ge_{c(v_k)} \mu(v_1).
\]
Since $p$ is the largest priority in the loop we have $c(v_i) \ge p$, so in particular these inequalities hold for the coarser $p$-order~$\ge_p$:
\[
\mu(v_1) >_p \mu(v_2) \ge_p \mu(v_3) \ge_p \cdots \ge_p \mu(v_1).
\]
\textit{i.e.} $\mu(v_1) >_p \mu(v_1)$, a contradiction.
We just proved that all cycles consistent with $\sigma$ have a largest even priority, which implies that $\sigma$ is indeed winning
for the parity objective.
\end{proof}

\begin{theorem}[Existence of signatures for parity games~\cite{EJ91}]\label{thm:signature}
For all parity games with $n$ vertices and $d$ priorities, 
there exists a signature $\mu : V \to [0,n]^{d/2} \cup \set{\top}$
such that for all $v \in V$, we have $\mu(v) \neq \top$ if and only if Eve wins from $v$.
\end{theorem}

The original proof is due to Emerson and Jutla~\cite{EJ91}.
In the remainder of this section we revisit Zielonka's algorithm with one objective in mind: 
obtaining an alternative proof of Theorem~\ref{thm:signature}.

\subsection*{Zielonka's algorithm}

We revisit the first algorithm constructed to solve parity games due to Zielonka~\cite{Zielonka98}, 
adapting ideas from~\cite{McNaughton93}.

The reader familiar with parity games may jump to the next section; this section does not contain any new results.
We hope that the mildly unusual presentation of Zielonka's algorithm can give the non-expert reader some insights into parity games,
and help reading the rest of the paper.

\vskip1em
We introduce some notations.
For a set of vertices $U \subseteq V$, we let $\Pre(U) \subseteq V$ be the set of vertices from which Eve can ensure to reach $U$ in one step:
\[
\begin{array}{lll}
\Pre(U) & = & \set{u \in V_E \mid \exists (u,v) \in E,\ v \in U} \\
        & \cup & \set{u \in V_A \mid \forall (u,v) \in E,\ v \in U}.
\end{array}
\]
We use $\overline{\Pre(U)}$ for the complement of $\Pre(U)$.
For a colour~$c$, the objective $\Reach(c)$ is satisfied by plays visiting some vertex of colour $c$ at least once, 
and $\Safe(c)$ by plays never visiting any vertex of colour $c$.

\vskip1em
Let us consider a parity game $\game$ with $d$ priorities.
We construct two recursive procedures, which take as input a (small variant of a) parity game with priorities in $[1,p]$ 
and two additional colours: $\set{\Win,\Lose}$, and output the winning set for Eve.
The vertices with colours $\Win$ or $\Lose$ are terminal: when reaching a terminal vertex, 
the game stops and one of the players is declared the winner.
Formally, the objective is
\[
\left( \Parity \cup \Reach(\Win) \right) \cap \Safe(\Lose).
\]

We write $V_p$ for the set of vertices of priority~$p$.

\subsubsection*{If the largest priority is even}

\begin{lemma}\label{lem:zielonka_even}
Consider a parity game $\game$ with priorities in $[1,p]$ with $p$ even.

Then $W_E (\left( \Parity \cup \Reach(\Win) \right) \cap \Safe(\Lose) )$ is the greatest fixed point of the operator 
\[
Y \mapsto W_E\left( 
\begin{array}{c} 
\Parity \cup \Reach \left[ \Win \cup (V_p \cap \Pre(Y)) \right] \\
\cap \\
\Safe \left[ \Lose \cup (V_p \cap \overline{\Pre(Y)}) \right] 
\end{array} \right).
\]
\end{lemma}
\noindent 
In words (for the sake of explanation, we assume that $\Win = \Lose = \emptyset$): 
$W_E(\Parity)$ is the largest set of vertices $Y$ such that Eve has a strategy ensuring that
\begin{itemize}
	\item either the priority $p$ is never seen, in which case the parity objective is satisfied with lower priorities,
	\item or the priority $p$ is seen, in which case Eve can ensure to reach $Y$ in one step.
\end{itemize}

\begin{proof}
We let $W$ denote 
\[
W_E (\left( \Parity \cup \Reach(\Win) \right) \cap \Safe(\Lose) ).
\]

The fact that $W$ is included in the greatest fixed point follows from the fact that it is itself a fixed point,
which is easy to check.

To prove that $W$ contains the greatest fixed point, we observe that any fixed point $Y$ is contained in $W$.
Indeed, if $Y$ is a fixed point, the strategy described above ensures parity: either it visits finitely many times $p$,
and then from some point onwards the parity objective is satisfied with lower priorities, or it visits infinitely many times $p$,
and then the parity objective is satisfied because $p$ is maximal and even.
Note that this strategy is positional, as disjoint union of two positional strategies,
one for vertices of priorities less than $p$ and the other for $V_p \cap \Pre(W)$.
\end{proof}

\begin{algorithm}
 \KwData{A parity game with priorities in $[1,p]$ with $p$ even and $\Win,\Lose$ two additional colours}

$Y_{-1} \leftarrow V$ ;

$k \leftarrow 0$ ;
     
\Repeat{$Y_k = Y_{k-1}$}{
$\Win_k \leftarrow V_p \cap \Pre(Y_{k-1})$ ;

$\Lose_k \leftarrow V_p \cap \overline{\Pre(Y_{k-1})}$ ;

$Y_k = W_E\left( 
\begin{array}{c}
\Parity \cup \Reach(\Win \cup \Win_k) \\ 
\cap \\
\Safe(\Lose \cup \Lose_k) 
\end{array}
\right)$ ;

$k \leftarrow k + 1$ ;}

\Return{$Y_k$}
\caption{The recursive algorithm when the largest priority is even.}
\label{algo:fixpoint_even}
\end{algorithm}

Algorithm~\ref{algo:fixpoint_even} fleshes out the fixed point computation described in Lemma~\ref{lem:zielonka_even},
which shows that it outputs 
\[
Y_k = W_E (\left( \Parity \cup \Reach(\Win) \right) \cap \Safe(\Lose)),
\]
with $\Win_k = Y_k \cap V_p$ and $\Lose_k = (V \setminus Y_k) \cap V_p$.

For each $k$ the computation of $Y_k$ is a recursive call: in the new game, 
vertices with priorities $p$ are marked terminal, 
and declared winning if in $\Pre(Y_{k-1})$ (\textit{i.e.} color $\Win$), losing otherwise (color $\Lose$). 
So in this game the priorities are in $[1,p-1]$.

\begin{figure*}[!ht]
\centering
\begin{subfigure}{.5\textwidth}
  \centering
  \includegraphics[width=.8\linewidth]{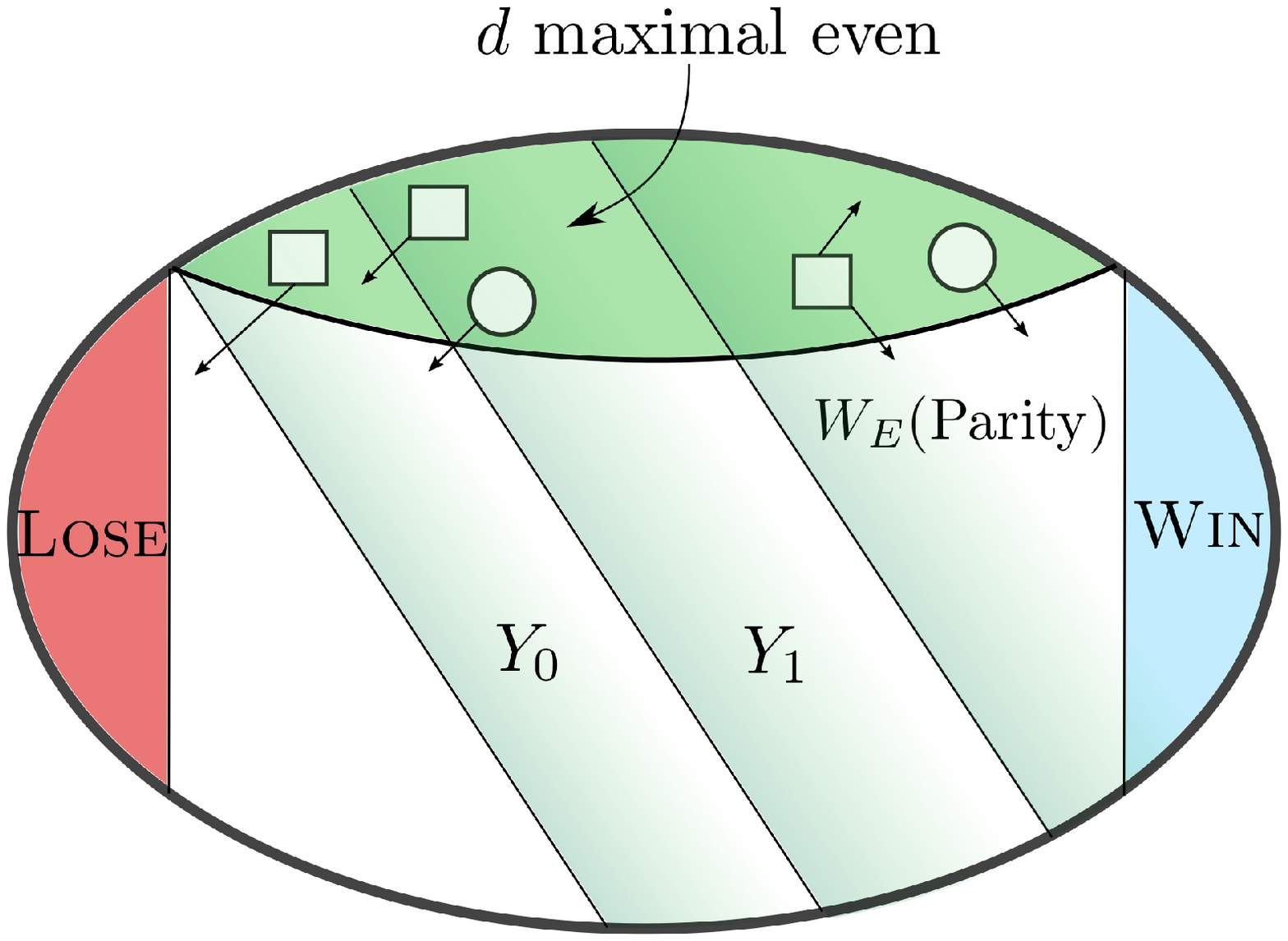}
  \label{fig:parity_even}
\end{subfigure}%
\begin{subfigure}{.5\textwidth}
  \centering
  \includegraphics[width=.8\linewidth]{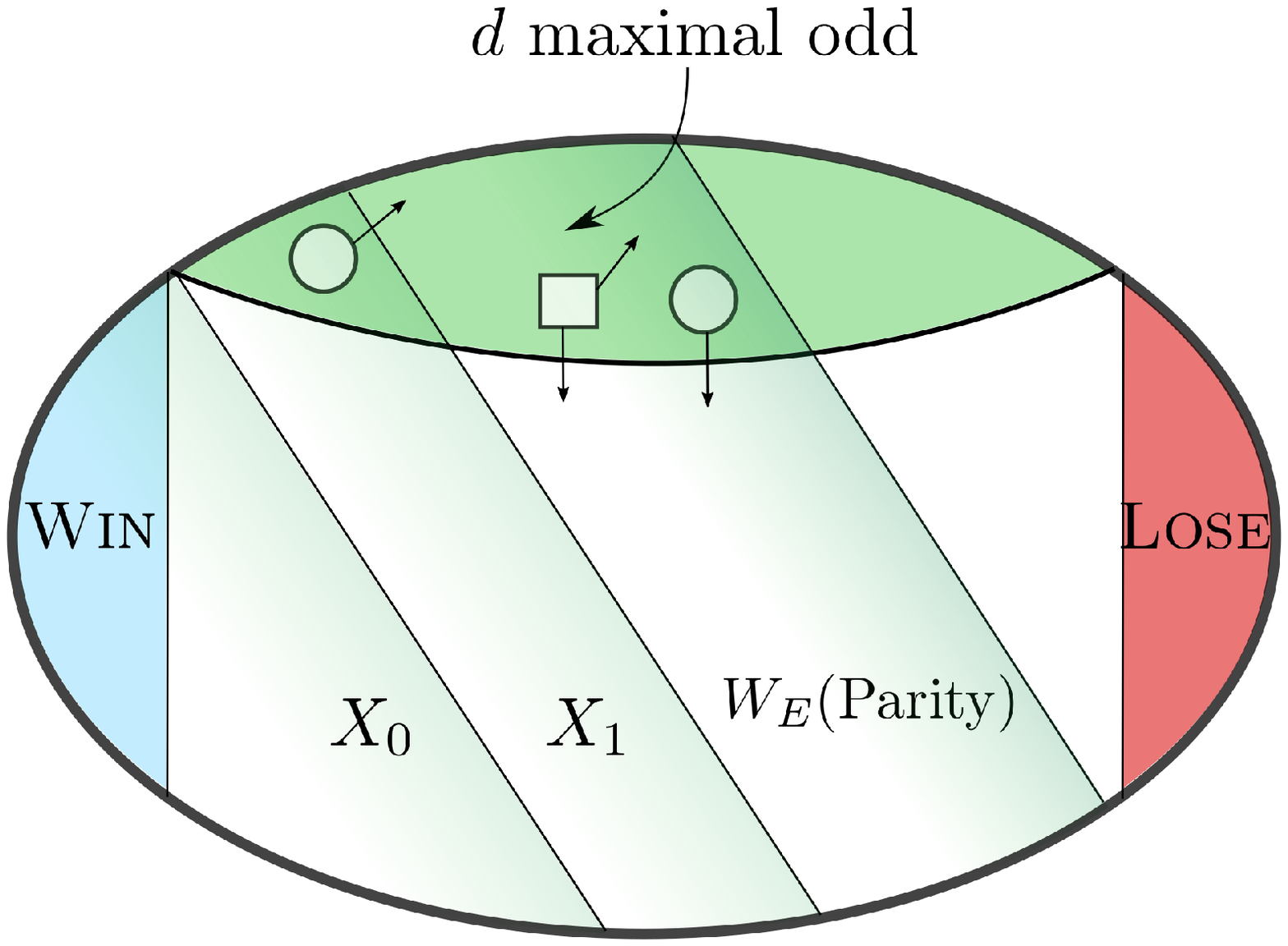}
  \label{fig:parity_odd}
\end{subfigure}
\caption{The two recursive procedures: even on the left and odd on the right.}
\label{fig:zielonka}
\end{figure*}

\subsubsection*{If the largest priority is odd}

\begin{remark}
At this point it is very tempting to say that the odd case is symmetric to the even case, swapping the role of the two players.
We do not take this road, because it requires assuming determinacy of parity games which we want to avoid in this presentation,
and obtain as a corollary.
It is also convenient to have the odd case spelled out for the construction of signatures.
\end{remark}

\begin{lemma}\label{lem:zielonka_odd}
Consider a parity game $\game$ with priorities in $[1,p]$ with $p > 1$ odd.

Then $W_E (\left( \Parity \cup \Reach(\Win) \right) \cap \Safe(\Lose) )$ is the least fixed point of the operator 
\[
X \mapsto W_E\left( 
\begin{array}{c} 
\Parity \cup \Reach \left[ \Win \cup (V_p \cap \Pre(X)) \right] \\
\cap \\
\Safe \left[ \Lose \cup (V_p \cap \overline{\Pre(X)}) \right] 
\end{array} \right).
\]
\end{lemma}

\begin{proof}
We let $W$ denote 
\[
W_E (\left( \Parity \cup \Reach(\Win) \right) \cap \Safe(\Lose) ).
\]

The fact that $W$ contains the least fixed point follows from the fact that it is itself a fixed point,
which is easy to check.

To prove that $W$ is included in the least fixed point is the interesting and non-trivial bit.
It follows from the observation that any fixed point $X$ contains $W$.
We show that 
\[
V \setminus X \subseteq W_A (\left( \Parity \cup \Reach(\Win) \right) \cap \Safe(\Lose) ) \subseteq V \setminus W.
\]
Note that here we are not relying on the determinacy of parity games: the second inclusion is very simple and always true,
it only says that Eve and Adam cannot win from the same vertex.

Indeed, if $X$ is a fixed point, from $V \setminus X$ Adam has a strategy ensuring that
\begin{itemize}
	\item either the priority $p$ is never seen, in which case the parity objective is violated with lower priorities,
	\item or the priority $p$ is seen, in which case Adam can ensure to reach $V \setminus X$ in one step.
\end{itemize}
This strategy violates parity: either it visits finitely many times $p$,
and then from some point onwards the parity objective is violated with lower priorities, 
or it visits infinitely many times $p$, and then the parity objective is violated because $p$ is maximal and odd.
\end{proof}

\begin{algorithm}
 \KwData{A parity game with priorities in $[1,p]$ with $p > 1$ odd and $\Win,\Lose$ two additional colours}

$X_{-1} \leftarrow \emptyset$ ; 

$k \leftarrow 0$ ;
     
\Repeat{$X_k = X_{k-1}$}{
$\Win_k \leftarrow V_p \cap \Pre(X_{k-1})$ ;

$\Lose_k \leftarrow V_p \cap \overline{\Pre(X_{k-1})}$ ;

$X_k = W_E\left( 
\begin{array}{c}
\Parity \cup \Reach(\Win \cup \Win_k) \\ 
\cap \\
\Safe(\Lose \cup \Lose_k) 
\end{array}
\right)$ ;

$k \leftarrow k + 1$ ;}

\Return{$X_k$}
\caption{The recursive algorithm when the largest priority is odd.}
\label{algo:fixpoint_odd}
\end{algorithm}

The base case $p = 1$ is easily dealt with by computing $\WE(\Reach(\Win) \cap \Safe(\Lose))$.
Zielonka's algorithm alternates greatest and least fixed point computations, in total $d-1$ of them.
Each of them computes subsets of the vertices, hence stabilises within at most $n$ steps.
A careful analysis gives a time complexity bound of $O(m \cdot (n/d)^d)$~\cite{Jurdzinski00}.

\subsubsection*{The construction of signatures}

We now analyse the structural decomposition unearthed by Zielonka's algorithm.
We fix a parity game $\game$.
For an odd priority $p$, consider the the non-decreasing sequence of sets of vertices 
\[
X_0(p) \subseteq X_1(p) \subseteq X_2(p) \subseteq \cdots 
\]
computed by running the algorithm with inputs $\game$ and
\[
\Win = \WE(\Parity) \cap V_{\ge p}\ ;\ \Lose = \WA(\Parity) \cap V_{\ge p}.
\]
We define a function $\mu : V \to [0,n]^{d/2} \cup \set{\top}$ as follows:
for $p$ an odd priority, $\mu(p)(v)$ is the smallest $k$ such that $v$ is in $X_k(p)$,
and $\top$ if it does not belong to any of these sets.

\begin{lemma}
The function $\mu$ defined above is a signature
such that for all $v \in V$, we have $\mu(v) \neq \top$ if and only if Eve wins from $v$.
\end{lemma}

\begin{proof}
We let $\sigma$ be the positional strategy constructed in the proof of Lemma~\ref{lem:zielonka_odd}.
Let $v \in V$ of priority $p$, we make two observations.
\begin{itemize}
	\item If $v \in X_k(p')$ with $p' > p$, the strategy $\sigma$ ensures to remain in $X_k(p')$ in the next step.
	\item If $p$ is odd and $v \in X_k(p)$, the strategy $\sigma$ ensures to reach $X_{k-1}(p)$ in the next step.
\end{itemize}
These two properties imply that $\mu$ is a signature.
The equivalence between $\mu(v) \neq \top$ and the fact that Eve wins from $v$ is a corollary of the correctness
of the algorithm given by Lemma~\ref{lem:zielonka_even} and Lemma~\ref{lem:zielonka_odd}.
\end{proof}



\section{A generic value iteration algorithm}
\label{sec:algorithm}

In this section, we define the notion of universal trees, and show how given a universal tree
one can construct a value iteration algorithm for parity games.
Both the small progress measure and the succinct progress measure algorithms are instances of this generic value iteration algorithm.

\subsection{Universal trees}
Let us fix two parameters $n$ and $h$.
The trees we consider have the following properties:
\begin{itemize}
	\item There are totally ordered, meaning that each node has a totally ordered set of children;
	\item They have a designated root and all leaves have depth exactly $h$.
\end{itemize}

We say that a tree embeds into another if the first one can be obtained by removing nodes from the second,
mapping root to root: in graph-theoretic terms, the first tree is a subgraph of the second.
We say that a tree $T$ is $(n,h)$-universal if all trees with at most $n$ leaves embed into $T$.
(Equivalently, it is enough to require that all trees with exactly $n$ leaves embed into $T$.)
An example of a $(n,h)$-universal tree is the complete tree of height $h$ with each node of degree $n$, 
it has $n^h$ leaves, as illustrated in Figure~\ref{fig:example_universal}.

\begin{figure*}[!ht]
\centering
\begin{subfigure}{.5\textwidth}
  \centering
  \includegraphics[width=.8\linewidth]{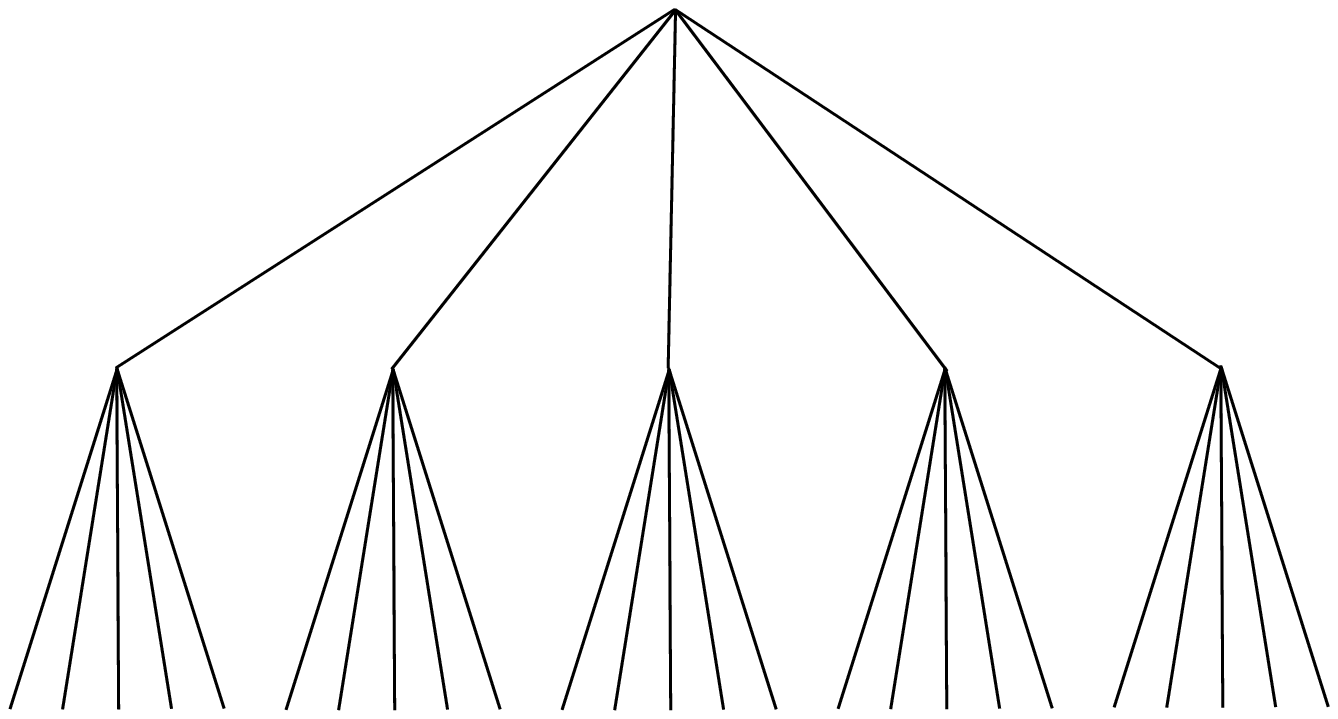}
  \label{fig:naive}
\end{subfigure}%
\begin{subfigure}{.5\textwidth}
  \centering
  \includegraphics[width=.8\linewidth]{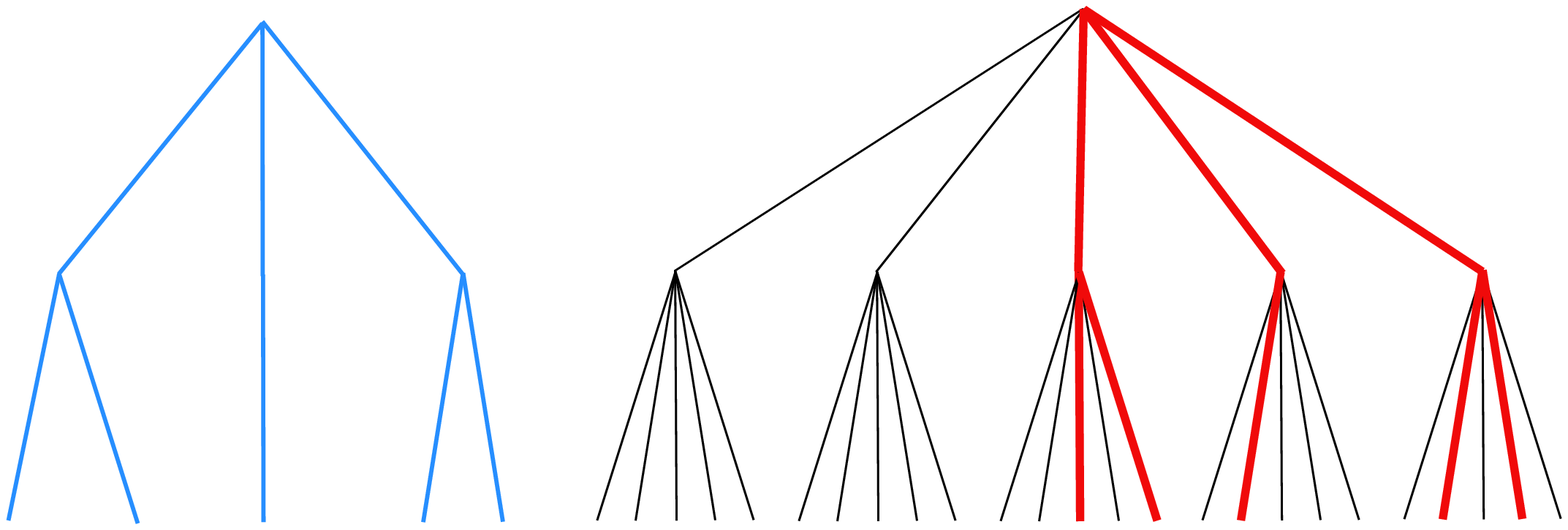}
  \label{fig:example_embedding}
\end{subfigure}
\caption{On the left, the naive $(5,2)$-universal tree with 25 leaves. 
On the right, a tree with $5$ leaves and height $2$, and one possible embedding into the naive universal tree.}
\label{fig:example_universal}
\end{figure*}

The size of a tree is the number of leaves it has.
We show in Figure~\ref{fig:tree_optimal} the smallest $(5,2)$-universal tree. It has $11$ leaves, which is less than the naive one ($25$ leaves).

\begin{figure}[!ht]
\centering
\includegraphics[scale=.4]{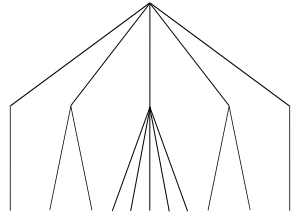}
\caption{The smallest $(5,2)$-universal tree has $11$ leaves.}
\label{fig:tree_optimal}
\end{figure}

\subsection{Signatures as trees}

The small progress measure algorithm casts the problem of constructing a signature as a least fixed point computation.
It assigns to each vertex a tuple in $[0,n]^{d/2}$ and updates the values of the vertices in order to satisfy the local constraints of signatures.
In other words the algorithm manipulates functions $\mu : V \to [0,n]^{d/2}$, which can equivalently seen as trees 
as illustrated in Figure~\ref{fig:example_signature}.
Each vertex is given by its path from the root, which has length $d/2$, and each direction is labeled by a number in $[0,n]$.
\begin{figure}[!ht]
\centering
\includegraphics[scale=.5]{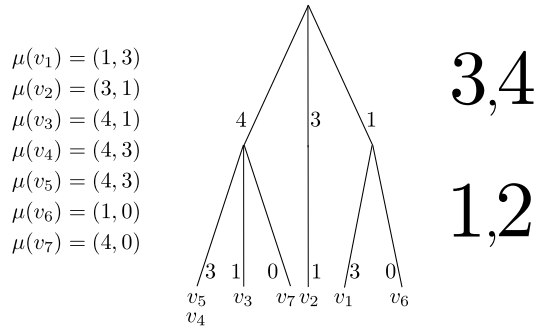}
\caption{A function $\mu : V \to [0,n]^{d/2}$ induces a tree. Here $n = 7$ and $d = 4$.}
\label{fig:example_signature}
\end{figure}

\vskip1em
The tree representation naturally induces the $p$-orders $\ge_p$. Indeed, indexing the levels from bottom to top by pairs of priorities
as in Figure~\ref{fig:example_signature},
whether $\mu(v) \ge_p \mu(v')$ can be read off the tree: it is equivalent to saying that the ancestor of $v$ at level $p$
is to the left of the ancestor of $v'$ at level $p$.
For instance in Figure~\ref{fig:example_signature} we have $\mu(v_3) >_1 \mu(v_7)$ but $\mu(v_3) =_3 \mu(v_7)$, and $\mu(v_2) >_2 \mu(v_1)$.

\vskip1em
Now, recall that the end goal of the algorithm is to construct a signature. 
A closer inspection at the definition of signatures reveals that the choice of values for the direction is immaterial:
the definition only uses the orders $\ge_p$. 
In other words, being a signature is a property of the underlying (totally ordered) tree,
and $[0,n]$ is just a total order among others.



\subsection{Existence of signatures for universal trees}


Theorem~\ref{thm:signature} considers functions $\mu : V \to [0,n]^{d/2} \cup \set{\top}$,
which as we explained can be seen as trees.
In the following theorem we fix a universal tree $T$ and we consider functions of the form
$\mu : V \to L(T) \cup \set{\top}$, where $L(T)$ is the set of leaves of $T$.

Such a function induces a set of orders on vertices called the $p$-orders: 
for $p \in [1,d]$ and $v,v'$ two vertices, we have $\mu(v) \ge_p \mu(v')$ if
the ancestor at level $p$ of $v$ is to the left of the ancestor at the same level of $v'$ 
(where levels are indexed as in Figure~\ref{fig:example_signature}).
The element $\top$ is the largest element for all $p$-orders $\ge_p$.

We extend the definition of signatures to functions $\mu : V \to L(T) \cup \set{\top}$,
using the exact same two properties which only depend upon the $p$-orders $\ge_p$.
Lemma~\ref{lem:signature_correctness} extends \textit{mutatis mutandis}:
indeed, the proof does not depend upon the choice of the underlying universal tree
but only on the $p$-orders.



\begin{theorem}\label{thm:signature_general}
For all parity games with $n$ vertices and $d$ priorities, 
for all $(n,d/2)$-universal tree $T$,
there exists a signature $\mu : V \to L(T) \cup \set{\top}$
such that for all $v \in V$, we have $\mu(v) \neq \top$ if and only if Eve wins from $v$.
\end{theorem}

\begin{proof}
Let $T$ be a $(n,d/2)$-universal tree $T$ and $\game$ be a parity game with $n$ vertices and $d$ priorities.
Thanks to Theorem~\ref{thm:signature} there exists a signature $\mu : V \to [0,n]^{d/2} \cup \set{\top}$
such that for all $v \in V$, we have $\mu(v) \neq \top$ if and only if Eve wins from $v$.
As explained in Figure~\ref{fig:example_signature} this induces a tree $t$ with at most $n$ leaves.
The crucial property is that $\mu$ and $t$ induce the same $p$-orders $\ge_p$ for all $p \in [1,d]$.
Since $T$ is $(n,d/2)$-universal, the tree $t$ embeds into $T$. 
Now, this induces a signature 
\[
\mu : V \to L(T) \cup \set{\top},
\]
since the definition of signatures only depends on the $p$-orders $\ge_p$.
\end{proof}

\subsection{The generic value iteration algorithm}
We construct a value iteration algorithm parameterised by the choice of a universal tree.
We fix $n$ the number of vertices and $d$ the number of priorities,
and a $(n,d/2)$-universal tree~$T$.
We let $\ell_{\min}$ denote the smallest leaf with respect to $\ge$, \textit{i.e.} the rightmost leaf of~$T$.
For $v \in V$ of priority $p$, define $\Lift_v(\mu) \in L(T)$ to be:
\begin{itemize}
	\item If $v \in \VE$, the smallest leaf $\ell$ with respect to $\ge_p$ such that there exists $(v,v') \in E$
	and $\ell \ge_p \mu(v')$, with a strict inequality if $p$ is odd;
	\item If $v \in \VA$, the smallest leaf $\ell$ with respect to $\ge_p$ such that for all $(v,v') \in E$,
	we have $\ell \ge_p \mu(v')$, with a strict inequality if $p$ is odd.
\end{itemize}
The definition of $\mu$ being a signature naturally reformulates in: for all $v \in V$, we have $\mu(v) = \Lift_v(\mu)$.

\begin{algorithm}[!ht]
 \KwData{A parity game with $n$ vertices and $d$ priorities.}

\For{$v \in V$}{
$\mu(v) \leftarrow \ell_{\min}$ ;
}
     
\Repeat{$\exists v \in V,\ \mu(v) \neq \Lift_v(\mu)$}{
Choose $v \in V$ such that $\mu(v) \neq \Lift_v(\mu)$ ;

$\mu(v) \leftarrow \Lift_v(\mu)$ ;}

\Return{$\mu$}
\caption{The generic value iteration algorithm.}
\label{algo:value_iteration}
\end{algorithm}

The algorithm is given in Algorithm~\ref{algo:value_iteration}, and its correctness follows from a lemma we present now,
an exact replica of Theorem 5 in~\cite{JL17}.
The operator $\Lift_v$ is extended to functions $V \to L(T) \cup \set{\top}$, updating the value of $v$
and leaving the other values unchanged.

\begin{lemma}
The set of all functions $V \to L(T) \cup \set{\top}$ is equipped with the pointwise order induced from $L(T) \cup \set{\top}$,
defining a finite complete lattice.
For all $v \in V$, the operator $\Lift_v$ is inflationary and monotone.
\end{lemma}

As explained in~\cite{JL17}, it follows from the lemma that from every $\mu : V \to L(T) \cup \set{\top}$, 
every sequence of applications of the operators $\Lift$ eventually
reaches the least simultaneous fixed point of all the operators $\Lift$ that is greater than or equal to $\mu$.
Hence we obtain the correctness of the generic value iteration algorithm.

\begin{theorem}
For all $n, d \in \N$ with $d$ even, for all $(n, d/2)$-universal tree $T$, 
for all parity games with $n$ vertices, $m$ edges, and $d$ priorities,
the value iteration algorithm over the tree $T$ outputs a signature $\mu$
such that for all $v \in V$, we have $\mu(v) \neq \top$ if and only if Eve wins from $v$.

Furthermore, the algorithm runs in time $O(m \log(n) \log(d) \cdot |T|)$, where $|T|$ is the number of leaves of $T$.
\end{theorem}

\subsubsection*{Complexity analysis}

The value iteration algorithm given above lifts a vertex $v$ at most $|T|$ many times, hence the total number of lifts is at most~$n |T|$.
This bound cannot be much improved: for instance a vertex of priority $1$ with a self-loop is evidently losing but the algorithm
will lift the vertex $|T|$ times to get this information.
Computing a lift for $v \in V$ can be performed in time $O(\text{deg}(v) \log(n) \log(d))$.
It follows that the complexity of the algorithm is proportional to the size of the underlying universal tree.

\subsubsection*{Two instances of the generic algorithm}

The small progress measure is an instance of the generic value iteration algorithm,
using the naive universal tree of size $n^h$ hence giving a running time in $n^{d/2 + O(1)}$.

\vskip1em
The succinct progress measure is an instance of the generic value iteration algorithm 
using a universal tree they construct in~\cite{JL17}. 
Indeed Lemma 1 in their paper exactly says that the (implicit) tree they construct is universal, 
by (inductively) constructing embeddings.
Their universal tree has quasipolynomial size (we elaborate on its construction in the next section),
hence the running time of the succinct progress measure algorithm is $n^{O(\log(d))}$.

We note that they additionally show that for their universal tree lifts can be performed in nearly linear space,
implying that the overall space complexity is nearly linear. 
This result is specific to the universal tree they construct and does not hold in general.


\section{Bounds on universal trees}
\label{sec:bounds}

We saw in the previous section that constructing universal trees gives value iteration algorithms for parity games,
and that the smaller the universal tree the better the time complexity.
We prove in this section upper and lower bounds on the size of universal trees.

\subsection{The (streamlined) succinct universal tree of Jurdzi{\'n}ski and Lazi{\'c}}

We present an inductive construction for succinct universal trees.
It is essentially the same as the construction of Jurdzi{\'n}ski and Lazi{\'c} in~\cite{JL17},
but the framework of universal trees allows us to avoid some rounding in the original construction,
hence a marginal improvement.

\begin{theorem}
There exists a $(n,h)$-universal tree with $f(n,h)$ leaves, where $f$ satisfies the following:
$$\begin{array}{lll}
f(n,h) & = & f(n,h-1) + f(\lfloor n/2 \rfloor,h) + f(n - 1 - \lfloor n/2 \rfloor,h), \\
f(n,1) & = & n, \\
f(1,h) & = & 1.
\end{array}$$
An upper bound is given by
\[
f(n,h) \le 2^{\lceil \log(n) \rceil} \binom{\lceil \log(n) \rceil + h - 1}{\lceil \log(n) \rceil}.
\]
\end{theorem}

\begin{proof}
To construct the $(n,h)$-universal tree $T$, let:
\begin{itemize}
	\item $T_\text{left}$ be a $(\lfloor n/2 \rfloor,h)$-universal tree;
	\item $T_\text{middle}$ be a $(n,h-1)$-universal tree;
	\item $T_\text{right}$ be a $(n - 1 - \lfloor n/2 \rfloor,h)$-universal tree.
\end{itemize}
We construct $T$ as in Figure~\ref{fig:smallest_tree_construction}.
More precisely, the children of the root is $T$ are, in order: the children of $T_\text{left}$, then
the root of $T_\text{middle}$, and then the children of $T_\text{right}$.

\begin{figure}[!ht]
\centering
\includegraphics[scale=.5]{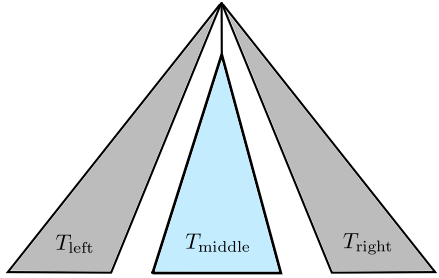}
\caption{The inductive construction.}
\label{fig:smallest_tree_construction}
\end{figure}

\vskip1em
We argue that $T$ is $(n,h)$-universal.
Consider a tree $t$ with $n$ leaves.
The question is where to cut in the middle, \textit{i.e.} which child of the root of $t$ gets mapped to the root of $T_\text{middle}$.
Let $v_1,\ldots,v_m$ be the children of the root of $t$, and let $n(v_i)$ be the number of leaves below $v_i$. 
Since $t$ has $n$ leaves, we have $n(v_1) + \cdots + n(v_m) = n$.
There exists a unique $k$ such that 
\[
\begin{array}{l}
n(v_1) + \cdots + n(v_{k-1}) \le \lfloor n/2 \rfloor, \text{ and } \\
n(v_1) + \cdots + n(v_k) > \lfloor n/2 \rfloor.
\end{array}
\]
For this choice of $k$ we have
\[
n(v_{k+1}) + \cdots + n(v_m) \le n - 1 - \lfloor n/2 \rfloor.
\]
To embed $t$ into $T$, we proceed as follows:
\begin{itemize}
	\item the tree rooted in $v_p$ has height $h-1$ and at most $n$ leaves, so in embeds into $T_\text{middle}$;
	\item the tree obtaining by restricting $t$ to all nodes to the left of $v_k$ has $\lfloor n/2 \rfloor$,
	so it embeds into $T_\text{left}$ by induction hypothesis;
	\item the tree obtaining by restricting $t$ to all nodes to the right of $v_k$ has $n - 1 - \lfloor n/2 \rfloor$,
	so it embeds into $T_\text{right}$ by induction hypothesis.
\end{itemize}
\end{proof}

\subsubsection*{Analysis of the function $f$}

Define $F(p,h) = f(2^p,h)$ for $p \ge 0$ and $h \ge 1$.
Then we have
$$\begin{array}{lll}
F(p,h) & \le & F(p,h-1) + 2 F(p-1,h), \\
F(p,1) & = & 2^p, \\
F(0,h) & = & 1.
\end{array}$$
To obtain an upper bound on $F$ we define $\overline{F}$ by
$$\begin{array}{lll}
\overline{F}(p,h) & = & \overline{F}(p,h-1) + 2 \overline{F}(p-1,h), \\
\overline{F}(p,1) & = & 2^p, \\
\overline{F}(0,h) & = & 1,
\end{array}$$
so that $F(p,h) \le \overline{F}(p,h)$.
Define the bivariate generating function 
\[
\F(x,y) = \sum_{p \ge 0, h \ge 1} \overline{F}(p,h) x^p y^h.
\]
Plugging the inductive equalities we obtain
\[
\F(x,y) = \frac{y}{1 - 2x - y},
\]
from which we extract that $\overline{F}(p,h) = 2^p \binom{p+h-1}{p}$, implying $F(p,h) \le 2^p \binom{p+h-1}{p}$.
Putting everything together we obtain
\[
f(n,h) \le 2^{\lceil \log(n) \rceil} \binom{\lceil \log(n) \rceil + h - 1}{\lceil \log(n) \rceil}.
\]
Note that this is very close and marginally better than the bound obtained in~\cite{JL17},
which is $2^{\lceil \log(n) \rceil} \binom{\lceil \log(n) \rceil + h + 1}{\lceil \log(n) \rceil}$.

\begin{corollary}
There exists an algorithm solving parity games in time 
\[
O\left(m n \log(n) \log(d) \cdot \binom{\lceil \log(n) \rceil + d/2 - 1}{\lceil \log(n) \rceil}\right).
\]
\end{corollary}

\subsection{Lower bounds on universal trees}
\begin{theorem}
Any $(n,h)$-universal tree has at least $g(n,h)$ leaves, where $g$ satisfies the following:
$$\begin{array}{lll}
g(n,h) & = & \sum_{\delta = 1}^n g(\lfloor n / \delta \rfloor,h-1), \\
g(n,1) & = & n, \\
g(1,h) & = & 1.
\end{array}$$
A lower bound is given by
\[
g(n,h) \ge \binom{\lfloor \log(n) \rfloor + h - 1}{\lfloor \log(n) \rfloor}.
\]
\end{theorem}

This lower bound shares some similarities with a result from Goldberg and Lifschitz~\cite{GL68},
which is for universal trees of a different kind: 
the height is not bounded and the children of a node are not ordered.

\begin{proof}
We proceed by induction. The bounds are clear for $h = 1$ or $n = 1$.

Let $T$ be a $(n,h)$-universal tree, and $\delta \in [1,n]$. 
We claim that the number of nodes at depth $h-1$ 
of degree greater to or larger than $\delta$ is at least $g(\lfloor n / \delta \rfloor,h-1)$.

\vskip1em
Let $T_\delta$ be the subtree of $T$ obtained by removing all leaves and all nodes at depth $h-1$
of degree less than $\delta$: the leaves of the tree $T_\delta$ have height exactly $h-1$.

We argue that $T_\delta$ is $(\lfloor n / \delta \rfloor,h-1)$-universal.
Indeed, let $t$ be a tree with $\lfloor n / \delta \rfloor$ leaves all at depth $h-1$.
To each leaf of $t$ we append $\delta$ children, yielding the tree $t_+$ which has $\lfloor n / \delta \rfloor \cdot \delta \le n$ leaves 
all at depth~$h$.
Since $T$ is $(n,h)$-universal, the tree $t_+$ embeds into $T$.
Observe that the embedding induces an embedding of $t$ into $T_\delta$,
since the leaves of $t$ have degree $\delta$ in $t_+$, hence are also in $T_\delta$.

\vskip1em
So far we proved that the number of nodes at depth $h-1$ 
of degree greater to or larger than $\delta$ is at least $g(\lfloor n / \delta \rfloor,h-1)$.
Now, note that the sum over $\delta \in [1,n]$ of the number of nodes at depth $h-1$ 
of degree greater to or larger than $\delta$ is a lower bound on the number of leaves,
which concludes.
\end{proof}

\subsubsection*{Analysis of the function $g$}

Define $G(p,h) = g(2^p,h)$ for $p \ge 0$ and $h \ge 1$.
Then we have
$$\begin{array}{lll}
G(p,h) & \ge & \sum_{k = 0}^p G(p-k,h-1), \\
G(p,1) & \ge & 1, \\
G(0,h) & = & 1.
\end{array}$$
To obtain a lower bound on $G$ we proceed similarly as for $F$.
We define $\overline{G}$ by
$$\begin{array}{lll}
\overline{G}(p,h) & = & \overline{G}(p,h-1) + \overline{G}(p-1,h), \\
\overline{G}(p,1) & = & 1, \\
\overline{G}(0,h) & = & 1,
\end{array}$$
so that $G(p,h) \ge \overline{G}(p,h)$.
Define the bivariate generating function 
\[
\G(x,y) = \sum_{p \ge 0, h \ge 1} \overline{G}(p,h) x^p y^h.
\]
Plugging the inductive equalities we obtain
\[
\G(x,y) = \frac{y}{1 - x - y},
\]
from which we extract that $\overline{G}(p,h) = \binom{p+h-1}{p}$, implying that $G(p,h) \ge \binom{p+h-1}{p}$.
Putting everything together we obtain
\[
g(n,h) \ge \binom{\lfloor \log(n) \rfloor + h - 1}{\lfloor \log(n) \rfloor}.
\]

The term $\binom{\lfloor \log(n) \rfloor + h - 1}{\lfloor \log(n) \rfloor}$ was analysed in depth in~\cite{JL17}
for various regimes relating $h$ and $n$.
It is quasipolynomial, inducing a quasipolynomial lower bound on the time complexity
of any instance of the generic value iteration algorithm for parity games.

\vskip1em
The upper and lower bounds do not match perfectly.
However, 
\[
\frac{f(n,h)}{g(n,h)} \le 2^{\lceil \log(n) \rceil} \frac{\lfloor \log(n) \rfloor + h}{\lfloor \log(n) \rfloor} = O(n h),
\]
\textit{i.e.} they are polynomially related, so it is fair to say that they \textit{almost} match.



\section{Perspectives}

We showed that the two versions of the value iteration algorithm, namely small progress measures and succinct progress measures,
can be seen as instances of a generic value iteration algorithm based on different universal trees.

By proving almost tight bounds on the size of universal trees essentially matching the succinct universal tree of Jurdzi{\'n}ski and Lazi{\'c},
we show that their result is optimal in this framework.
The bounds are not tight; it would be satisfying to sharpen the lower bound. 
We conjecture that the succinct universal tree we construct in this paper is actually optimal, meaning that there exist no smaller universal tree.

\vskip1em
How to proceed with the quest for a polynomial time algorithm for solving parity games?
The other quasipolynomial time algorithm due to Calude et al~\cite{CJKL017} does not fit the framework we introduce here,
hence is not subjected to the quasipolynomial lower bound proved in this paper.

\vskip1em
Boja{\'n}czyk and Czerwi{\'n}ski~\cite{BC17} offer an interesting perspective on the algorithm of Calude et al, showing that
it provides a solution to the following separation problem.

\vskip1em
We consider infinite words over the alphabet $V$.
A cycle is a word $v \cdots v$. It is even if the largest priority is even, and odd otherwise. 
We define two languages:
\[
\begin{array}{l}
\AllEvenCycles = \set{ \play \in V^\omega \mid \text{all cycles in } \play \text{ are even}}, \\
\AllOddCycles = \set{ \play \in V^\omega \mid \text{all cycles in } \play \text{ are odd}}.
\end{array}
\]
We look at deterministic safe automata: all states are accepting, a word is rejected only if there exists no run for it.
Such automata recognise exactly the set of topologically closed languages over infinite words.
 
The separation problem reads: construct a deterministic safe automaton recognising a language $L \subseteq V^\omega$ such that
\begin{itemize}
	\item $\AllEvenCycles \subseteq L$;
	\item $L \cap \AllOddCycles = \emptyset$,
\end{itemize}
as illustrated in Figure~\ref{fig:separation}.

\begin{figure}[!ht]
\centering
\includegraphics[scale=.3]{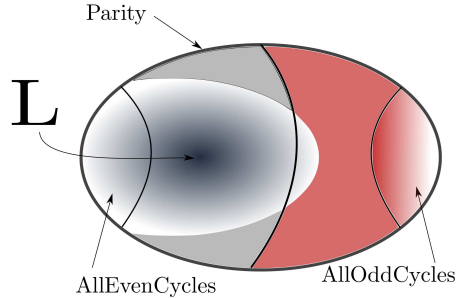}
\caption{The separation problem.}
\label{fig:separation}
\end{figure}

\begin{lemma}
If $L$ is a solution to the separation problem, then the winning regions of $\Parity$ and $L$ coincide.
\end{lemma}

Consequently, solving the parity game is equivalent to solving a safety game with $n \cdot |L|$ vertices
and $m \cdot |L|$ edges, where $|L|$ is the number of states of a deterministic automaton recognising $L$. 
Since solving a safety game can be done in linear time, more precisely in $O(m)$, 
this gives an algorithm for solving parity games whose running time is $O(m \cdot |L|)$.

\begin{proof}
This relies on the positional determinacy of parity games. 
A positional strategy for Eve ensuring $\Parity$ also ensures $\AllEvenCycles$, hence $L$. 
Conversely, a positional strategy for Adam ensuring the complement of $\Parity$ also ensures $\AllOddCycles$, hence the complement of $L$.
\end{proof}

Boja{\'n}czyk and Czerwi{\'n}ski~\cite{BC17} cast the data structure constructed in the algorithm of Calude et al~\cite{CJKL017}
as a solution of the separation problem.

\begin{theorem}[\cite{BC17}]
There exists a deterministic safe automaton solving the separation problem with $n^{O(\log(d))}$ states.
\end{theorem}

The next question is then: can we construct smaller solutions to the separation problem,
or can we prove lower bounds?

\section*{Acknowledgments}

The notion of universal trees was hinted at me by Marcin Jurdzi{\'n}ski and Ranko Lazi{\'c}.
They largely contributed to the making of this paper, and I thank them for their support.
I am very grateful to Albert Atserias for pointing out to me the literature on universal graphs,
Amos Korman for digging into the connection with distance labelings on trees,
Pawe{\l} Gawrychowski for discussions on lower bounds for universal trees,
and {\'E}lie de Panafieu for his expertise on combinatorial analysis.

\bibliography{bib}

\bibliographystyle{alpha}

\end{document}